\documentclass[a4paper]{article}
\usepackage[margin=1.25in]{geometry}
\usepackage{graphicx}
\usepackage{xcolor}

\usepackage{amsmath,amsthm,amssymb,amsfonts,tikz,tkz-euclide,mathtools}
\usepackage{authblk}
\usetikzlibrary{patterns.meta,decorations.pathmorphing,tikzmark,positioning,arrows,arrows.meta}
\usepackage[T1]{fontenc}
\usepackage{hyperref}
\usepackage[capitalize]{cleveref}

\newtheorem{theorem}{Theorem}
\theoremstyle{definition}

\newcommand{\A}{\mathcal{A}}
\newcommand{\G}{\mathcal{G}}

\definecolor{myblue}{HTML}{0088cc}
\definecolor{myorange}{HTML}{f26924}

\newcommand{\Rand}{{\operatorname{rand}}}
\newcommand{\Det}{{\operatorname{det}}}

\newcommand{\PP}{\mathcal{P}}
\newcommand{\RR}{\mathcal{R}}

\newcommand{\id}{\operatorname{id}}

\newcommand{\nat}{\mathbb{N}}

\newcommand{\inSigma}{\Sigma_{\operatorname{in}}}
\newcommand{\outSigma}{\Sigma_{\operatorname{out}}}
\newcommand{\inLbl}{\lambda_{\operatorname{in}}}
\newcommand{\outLbl}{\lambda_{\operatorname{out}}}

\newcommand{\Anorm}{\A_\text{norm}}
\newcommand{\Aexist}{\A}
 
\hypersetup{
  colorlinks = true,
  urlcolor = myblue,
  linkcolor = myblue,
  citecolor = myblue,
  pdftitle = {Distributed derandomization revisited},
  pdfauthor = {Sameep Dahal, Francesco d'Amore, Henrik Lievonen, Timothé Picavet, Jukka Suomela}
}

\newenvironment{mycover}
{\list{}{\listparindent 0pt
        \itemindent    \listparindent
        \leftmargin    0cm
        \rightmargin   0cm
        \parsep        0pt}\raggedright
    \item\relax}
{\endlist}

\newenvironment{myabstract}
{\list{}{\listparindent 1.5em\itemindent    \listparindent
        \leftmargin    0cm
        \rightmargin   0cm
        \parsep        0pt}\item\relax}
{\endlist}

\newcommand{\myaff}[1]{\,$\cdot$\, {\small #1}\par\smallskip}
\newcommand{\fakeparagraph}[2]{\par\noindent\textbf{#1}\hspace{1em}#2}

\begin{document}

\begin{mycover}
    {\huge\bfseries Distributed derandomization revisited}

    \bigskip
    \bigskip
    \bigskip
    \textbf{Sameep Dahal}
    \myaff{Aalto University, Finland}

    \textbf{Francesco d'Amore}
    \myaff{Aalto University, Finland}

    \textbf{Henrik Lievonen}
    \myaff{Aalto University, Finland}

    \textbf{Timothé Picavet}
    \myaff{Aalto University, Finland
    \,$\cdot$\,
    ENS de Lyon, France}

    \textbf{Jukka Suomela}
    \myaff{Aalto University, Finland}
\end{mycover}
\bigskip
\begin{myabstract}
    \fakeparagraph{Abstract.}
    One of the cornerstones of the distributed complexity theory is the derandomization result by Chang, Kopelowitz, and Pettie [FOCS 2016]: any randomized LOCAL algorithm that solves a locally checkable labeling problem (LCL) can be derandomized with at most exponential overhead.
    The original proof assumes that the number of random bits is bounded by some function of the input size. We give a new, simple proof that does not make any such assumptions---it holds even if the randomized algorithm uses infinitely many bits. While at it, we also broaden the scope of the result so that it is directly applicable far beyond LCL problems.
\end{myabstract}
\bigskip
\bigskip

\section{Introduction}

\paragraph{Distributed derandomization.}

A long line of recent work has led to a near-complete understanding of the distributed computational complexity of \emph{locally checkable labeling problems} (LCLs) \cite{suomela2020landscape}. These are graph problems that can be defined by giving a finite list of feasible local neighborhoods \cite{naor-stockmeyer1995}; for example, $c$-coloring in graphs of maximum degree $\Delta$ (for some fixed $c$ and $\Delta$) is an LCL problem.

We are in particular interested in the round complexity of LCLs in two standard models of distributed computing: deterministic and randomized versions of the LOCAL model \cite{linial,peleg_book}. One of the cornerstones of the distributed complexity theory is the derandomization result by Chang, Kopelowitz, and Pettie \cite[Theorem 3.1]{shattering}:

\begin{theorem}[Chang, Kopelowitz, and Pettie]\label{thm:main}
Let $\A_\Rand$ be a randomized LOCAL algorithm that solves an LCL problem $\PP$ in $T_\Rand(n)$ communication rounds in $n$-node graphs with probability at least $1-1/n$. Then there is a deterministic LOCAL algorithm $\A_\Det$ that solves $\PP$ in $T_\Det(n)$ rounds, where
$T_\Det(n) = T_\Rand\bigl(2^{n^2}\bigr)$.
\end{theorem}

But what do we mean, precisely, when we say that $\A_\Rand$ is a randomized algorithm in the LOCAL model? Chang, Kopelowitz, and Pettie \cite{shattering} assume that there is some upper bound $r(n)$ on the number of random bits used by a node. This is a non-standard definition; while many reasonable algorithms naturally satisfy this, formally speaking it is not compatible with e.g.\ a randomized algorithm in which each node picks a number from a geometric distribution by repeated Bernoulli trials. All other results that build on \cref{thm:main} are also influenced by this assumption; the foundations of the field are on a bit shaky ground.

\paragraph{New result: unbounded randomness.}

In this short note we prove a stronger version of \cref{thm:main}. Our proof does not need to assume anything about the number of random bits consumed by a node. Hence, we can now safely conclude that all corollaries of \cref{thm:main} also hold in the standard randomized LOCAL model, in which local computation---including the number of random bits generated---is unbounded.

Similar to \cite{shattering}, we assume that $n$ and $T_\Rand$ (or sufficiently tight bounds on them) are known. Similar to \cite{shattering}, the proof is constructive and $\A_\Det$ is a uniform, computable, deterministic algorithm. The only difference is that we assume less about $\A_\Rand$.

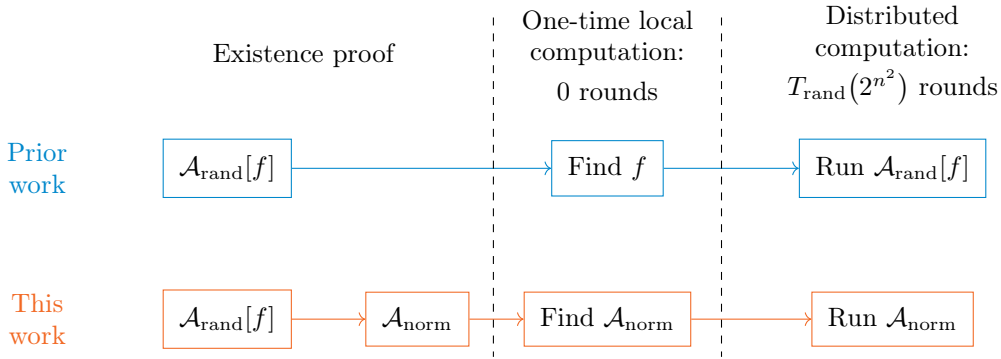
\begin{figure}
    \centering
    \begin{tikzpicture}[every node/.style={draw,rectangle,align=center,inner sep=6}]
        \node[draw=none,text=myblue] (1) at (0.5, 2) {Prior\\work};
        \node[draw=none,text=myorange] (2) at (0.5, 0) {This\\work};
        \node[draw=myblue] (3a) at (3, 2) {$\A_\Rand[f]$};
        \node[draw=myorange] (3b) at (3, 0) {$\A_\Rand[f]$};
        \node[draw=myorange] (5) at (5.5, 0) {$\A_\text{norm}$};
        \node[draw=myblue] (4) at (8, 2) {Find $f$};
        \node[draw=myorange] (6) at (8, 0) {Find $\A_\text{norm}$};
        \node[draw=myblue] (7) at (11.75, 2) {Run $\A_\Rand[f]$};
        \node[draw=myorange] (8) at (11.75, 0) {Run $\A_\text{norm}$};
        \draw[draw=myblue,->] (3a) to (4);
        \draw[draw=myorange,->] (3b) to (5);
        \draw[draw=myorange,->] (6) to (8);
        \draw[draw=myorange,->] (5) to (6);
        \draw[draw=myblue,->] (4) to (7);
        \draw[dashed] (6.5, 4) to (6.5,-0.5);
        \draw[dashed] (9.5, 4) to (9.5,-0.5);
        \node[draw=none] at (4, 3.5) {Existence proof};
        \node[draw=none] at (8, 3.5) {One-time local\\computation:\\[3pt]0 rounds};
        \node[draw=none] at (11.75, 3.5) {Distributed\\computation:\\[3pt]$T_\Rand\bigl(2^{n^2}\bigr)$ rounds};
    \end{tikzpicture}
    \caption{Proof strategy in this work and prior work~\cite{shattering}.}
    \label{fig:strategy}
\end{figure}

\paragraph{Key new ideas.}

Exactly like Chang, Kopelowitz, and Pettie \cite[Theorem 3.1]{shattering}, we start by defining $N = 2^{n^2}$. Then even though we are working in an $n$-node graph, we lie to $\A_\Rand$ that we have a graph with $N$ nodes. The running time increases to $T_\Rand(N)$, but the success probability improves to $1 - 1/N$, which is large enough to show that there exists a mapping $f$ from unique identifiers to random bits that works for every $n$-node graph.

At this point our paths deviate---see \Cref{fig:strategy} for an illustration. In \cite[Theorem 3.1]{shattering}, $\A_\Det$ is constructed as follows: Each node checks each possible mapping $f$, and picks the first one that works for every $n$-node graph; then $\A_\Det$ simply simulates $\A_\Rand$ with random bits from $f$. This is where they make use of bounded randomness: for a fixed $n$ there are only finitely many possible functions $f$ to check.

We proceed as follows---instead of looking at the \emph{internal behavior} of the algorithm we look at its \emph{external behavior}:
\begin{enumerate}
    \item Since a good mapping $f$ exists, we could in principle hard-code this specific mapping to obtain a deterministic algorithm $\Aexist = \A_\Rand[f]$. At this point we merely know that $\Aexist$ exists---this step is non-constructive, and $\Aexist$ might not even be computable.
    \item However, any deterministic LOCAL algorithm can be represented in a \emph{normal form} as a function $\Anorm$ that maps each possible $T_\Rand(N)$-radius neighborhood to a local output. Since $\Aexist$ exists, we know that such a function $\Anorm$ also exists and solves $\PP$ correctly in all $n$-node graphs.
    \item Now $\A_\Det$ simply finds the first valid $\Anorm$, and then simulates $\Anorm$.
\end{enumerate}
This way we can construct a computable, uniform, deterministic algorithm $\A_\Det$ even if we merely know that $\A_\Rand$ exists, and even if $\A_\Rand$ is non-computable or non-uniform.

\paragraph{Two extensions.}

While \cref{thm:main} was originally presented for LCL problems, our new proof works for a broader class of problems: we show how to handle labeling problems that are defined component-wise. The proof is given in \cref{sec:proof}; \cref{thm:main} then follows as a special case.

We also briefly discuss in \cref{sec:connected} one extension: how to derandomize algorithms that are only guaranteed to work in connected graphs. A bit more care is needed when we lie about the number of nodes in that case.

\section{Preliminaries}
Let \(G = (V,E)\) denote a simple undirected graph.
For any two nodes \(u,v \in V\), we denote their distance by \(d(u,v)\), i.e., the number edges in a shortest path connecting \(u\) to \(v\); if such path does not exist, then \(d(u,v) = +\infty\). 
Furthermore, by \(\deg({v})\) we denote the degree of \(v\), i.e., the number of incident edges.

\paragraph{LOCAL model.}
Let \(G = (V,E)\) be any graph with \(n\) nodes. 
In the \emph{deterministic LOCAL model}, each node $v \in V$ is given a unique identifier \(\id(v) \in \{1,2, \dots, n^c\}\) for some constant \(c \ge 1\). 
The initial knowledge of a node consists of its own identifier, its degree,  the number of nodes \(n\) and (possibly) an input label.
Each node runs the same algorithm and computation proceeds in synchronous rounds.
In each round, nodes send messages of arbitrary size to their neighbors, then receive some messages, and then perform local computations of arbitrary complexity.
After some number of rounds, a node must terminate its computation and decide on its local output.
The running time (or complexity) of a distributed algorithm is defined as the number of rounds needed by all nodes to decide the local output.

In the \emph{randomized LOCAL model}, each node is also given access to an infinite random bit stream, and the bit streams of the nodes are mutually independent.
We say that an algorithm is \emph{uniform} if the size of the description of the algorithm does not depend on $n$.

For any fixed locality~$T$, the LOCAL model can also be viewed as a mapping from each radius-$T$ neighborhood~$N_T[v]$ of each node~$v$ to a local output.
Here by $N_T[v]$ we mean the graph~$(V',E')$, where $V' \subseteq V$ is the set of all nodes~$u \in V(G)$ with $d(v, u) \leq T$ and $E'$ is the set of edges $\{s,t\} \in E$ with $d(v,s)\leq T-1$ and $d(v,t)\leq T$. Each node of $N_T[v]$ is also labeled with its original degree~$\deg(u)$, unique identifier~$\id(u)$, local input, and---for randomized algorithms---its stream of random bits.
This is exactly the information node~$v$ can gather in $T$ rounds.

\paragraph{Labeling problems.}
Let \(\inSigma\) be a finite set of input labels and \(\outSigma\) be a finite set of admissible output labels.
An \emph{input labeling} of a graph $G = (V,E)$ is a function $\inLbl\colon V \to \inSigma$, and an \emph{output labeling} is a function $\outLbl\colon V \to \outSigma$.
A \emph{labeling problem} \(\PP\) specifies for each graph and each input labeling a set of feasible output labelings.

We say that $\PP$ is a \emph{component-wise verifiable problem} if for each graph $G$ and each connected component $C$ of $G$, the set of valid output labelings restricted to $C$ only depends~on~$C$.

Let $r \in \nat$ be a constant. We say that $\PP$ is a \emph{locally verifiable problem} with verification radius $r$ if for each graph $G$ and each node $v$ of $G$, the set of valid output labelings restricted to $N_r[v]$ only depends on $N_r[v]$.

We note that LCL problems \cite{naor-stockmeyer1995} are a special case of locally verifiable problems with a constant bound on the degree of the nodes.
Locally verifiable problems are in turn a special case of component-wise verifiable problems.

\section{Main result}\label{sec:proof}

We give the derandomization result directly for component-wise verifiable problems; \cref{thm:main} then follows as a corollary.

\begin{theorem}\label{thm:main-new}
    Let $\A_\Rand$ be a randomized LOCAL algorithm that solves a component-wise verifiable problem $\PP$ in $T_\Rand(n)$ communication rounds in $n$-node graphs with probability at least $1-1/n$. Then there is a deterministic LOCAL algorithm $\A_\Det$ that solves $\PP$ in $T_\Det(n)$ rounds, where $T_\Det(n) = T_\Rand\bigl(2^{n^2}\bigr)$.
\end{theorem}

\begin{proof}
    Consider any sufficiently large $n$, and let $N = 2^{n^2}$.
    In what follows, we lie to algorithm $\A_\Rand$ that the input graph consists of $N$ nodes.
    Hence, it runs in time $T := T_\Rand(N) = T_\Det(n)$ and succeeds with probability $1 - 1/N$.

    Let $\RR_n= \{f \colon \{0,1\}^{c \log n} \to \{0,1\}^\mathbb{N}\}$ be the family of all possible assignments of random bits streams to unique identifiers.
    For $f \in \RR_n$, we write $\A_\Rand[f]$ to denote the \emph{deterministic} LOCAL algorithm in which node $v$ runs $\A_\Rand$ but uses $f(\id(v))$ as its random bit stream.
    Note that $\A_\Rand$ is equivalent to the following process: choose $f \in \RR_n$ uniformly at random and apply $\A_\Rand[f]$.

    Let $\G_n$ be the set of all possible inputs $(G, \id, \inLbl)$, where $G$ is an $n$-node graph, $\id$ is a unique identifier assignment, and $\inLbl$ is an input labeling.
    We know that
    \begin{equation*}
        |\G_n| \leq 2^{\binom{n}{2}}\cdot 2^{c n \log n}\cdot |\inSigma|^n < N = 2^{n^2}
    \end{equation*}
    for a large enough $n$.
    We say that $f$ is \emph{good} if $\A_\Rand[f]$ outputs a valid solution for \emph{every} input in family $\G_n$.

    Now, we show there exists a good $f$. Let $F$ be a uniform random variable over $\RR_n$. Then
    \[
        \Pr(F\text{ is bad}) \leq \sum_{G \in \G_n} \Pr(\A_\Rand[F]\text{ fails on }G)
        =\sum_{G \in \G_n} \Pr(\A_\Rand\text{ fails on }G)
        \leq \frac{|\G_n|}{N} < 1.
    \]
    Therefore, $\Pr(F \text{ is good}) > 0$. Hence, there exists a good function; let $f$ be any such function.
    Thus, there is a deterministic algorithm $\Aexist = \A_\Rand[f]$ that solves $\PP$ on all inputs in $\G_n$ in at most $T$ rounds.

    Any deterministic $T$-round algorithm in the LOCAL model defines a mapping $\Anorm$ from radius-$T$ neighborhoods to local outputs. Conversely, such a mapping $\Anorm$ can be interpreted as a $T$-round algorithm. Furthermore, for a fixed $n$, there are only finitely many such mappings.

    Now $\A_\Det$ works as follows:
    Given $n$, each node first enumerates all candidate mappings $\Anorm$ in lexicographic order,
    checks if $\Anorm$ solves $\PP$ for every $\G_n$, and stops once the first such $\Anorm$ is found.
    Then $\A_\Det$ uses $T$ rounds so that each node $v$ learns its radius-$T$ neighborhood $N_T[v]$, and finally each node applies mapping $\Anorm$ to $N_T[v]$ to determine its local output.
\end{proof}

\section{Technicality: connected graphs}\label{sec:connected}

In the proof of \cref{thm:main-new}, a key step was that we lied about $n$. The algorithm cannot catch us lying, as the $n$-node input graph $G$ is indistinguishable from some hypothetical $N$-node input graph $G'$ in which one connected component is isomorphic to $G$. As $\PP$ was assumed to be component-wise verifiable, an algorithm that succeeds globally in $G'$ also has to succeed locally when restricted to $G$.

The proof heavily exploited graphs that may consist of multiple connected components. In this section we briefly note that this is \emph{not necessary}. We can prove the following version of \cref{thm:main-new} that holds even if $\A_\Rand$ only works correctly in connected graphs. However, component-wise problems are too broad class of problems in this case, and we consider locally verifiable problems instead:

\begin{theorem}\label{thm:main-conn}
    Let $\A_\Rand$ be a randomized LOCAL algorithm that solves a locally verifiable problem $\PP$ in $T_\Rand(n)$ communication rounds in $n$-node connected graphs with probability at least $1-1/n$. Then there is a deterministic LOCAL algorithm $\A_\Det$ that solves $\PP$ in $O(T_\Det(n))$ rounds, where
    $T_\Det(n) = T_\Rand\bigl(2^{n^2}\bigr)$.
\end{theorem}

\begin{proof}
    Let $t = T_\Det(n) + r$, where $r$ is the verification radius of problem $\PP$.
    In algorithm $\A_\Det$, each node $v$ first explores its radius-$t$ neighborhood to determine if the entire input graph $G$ is contained in $N_t[v]$. If yes, we spend another $t$ rounds to inform all nodes about $G$. In this case all nodes have learned $G$, and we can solve $\PP$ by brute force and stop.

    Otherwise, we can proceed as we did in the proof of \cref{thm:main-new}. We can now safely lie about $N$. To see this, assume that $\A_\Rand$ fails in some $n$-node graph $G$ with probability more than $n/N$ if we lie that $G$ has $N$ nodes. Then the algorithm also has to fail locally in the radius-$r$ neighborhood of some node $v$ with probability more than $1/N$. Now it is possible to construct an $N$-node graph $G'$ with node $v'$ such that radius-$t$ neighborhood of $v$ in $G$ is isomorphic to the radius-$t$ neighborhood of $v'$ in $G'$ (here we exploit the fact that radius-$t$ neighborhood of $v$ does not contain the entire graph $G$). As radius-$t$ neighborhoods of $v$ and $v'$ agree, and the running time of $\A_\Rand$ is $t-r$ rounds, the output distributions of $N_r[v]$ and $N_r[v']$ also agree. Now it follows that $\A_\Rand$ fails locally in the radius-$r$ neighborhood of $v'$ in $G'$ with probability more than $1/N$, and hence it also fails globally in $G'$ with probability more than $1/N$, which is a contradiction with the assumption that $\A_\Rand$ solves $\PP$ in connected $N$-node graphs with probability at least $1-1/N$.

    Now as long as we choose a large enough $n$ such that $|\G_n| < N/n$, the rest of the proof of \cref{thm:main-new} goes through.
\end{proof}

\bibliographystyle{alphaurl}
\bibliography{bibliography}

\begin{thebibliography}{CKP19}

\bibitem[CKP19]{shattering}
Yi-Jun Chang, Tsvi Kopelowitz, and Seth Pettie.
\newblock An exponential separation between randomized and deterministic
  complexity in the local model.
\newblock {\em SIAM Journal on Computing}, 48(1):122--143, 2019.
\newblock \href {https://doi.org/10.1137/17M1117537}
  {\path{doi:10.1137/17M1117537}}.

\bibitem[Lin92]{linial}
Nathan Linial.
\newblock Locality in distributed graph algorithms.
\newblock {\em SIAM Journal on Computing}, 21(1):193--201, 1992.
\newblock \href {https://doi.org/10.1137/0221015} {\path{doi:10.1137/0221015}}.

\bibitem[NS95]{naor-stockmeyer1995}
Moni Naor and Larry~J. Stockmeyer.
\newblock What can be computed locally?
\newblock {\em SIAM Journal on Computing}, 24(6):1259--1277, 1995.
\newblock \href {https://doi.org/10.1137/S0097539793254571}
  {\path{doi:10.1137/S0097539793254571}}.

\bibitem[Pel00]{peleg_book}
David Peleg.
\newblock {\em Distributed Computing: A Locality-Sensitive Approach}.
\newblock SIAM, 2000.

\bibitem[Suo20]{suomela2020landscape}
Jukka Suomela.
\newblock Landscape of locality (invited talk).
\newblock In {\em Proc.\ SWAT}, 2020.
\newblock URL: \url{https://jukkasuomela.fi/landscape-of-locality/}, \href
  {https://doi.org/10.4230/LIPIcs.SWAT.2020.2}
  {\path{doi:10.4230/LIPIcs.SWAT.2020.2}}.

\end{thebibliography}

\end{document}